\documentclass{article} 
\usepackage{amsthm}
\usepackage{amsmath}

\usepackage{mathtools}
\usepackage{hyperref}

\usepackage{amssymb}
\usepackage{bbm}
\usepackage{xcolor}
\usepackage{comment}
\usepackage{tikz, tikz-3dplot}
\usepackage[T1]{fontenc}

\usetikzlibrary{arrows}
\pgfarrowsdeclarecombine*{spaced stealth'}{spaced stealth'}{stealth'}{stealth'}{space}{space}
\tikzset{>=spaced stealth'}

\newcommand*{\T}{^{\mkern-1.5mu\mathsf{T}}}

\newcommand{\N}{\mathbb{N}}
\newcommand{\Z}{\mathbb{Z}}
\newcommand{\ZZ}{\Z_{\ge 0}}
\newcommand{\R}{\mathbb{R}}
\newcommand{\RR}{\R_{\ge 0}}

\newcommand{\mIP}{mod-IP($a$)}

\usepackage{a4}
\usepackage{todonotes} 
\usepackage{enumerate} 
\usepackage{mathrsfs}

\usepackage[capitalize]{cleveref}

\newcommand{\defproblem}[3]{
	\vspace{2mm}
	\vspace{1mm}
	\noindent\fbox{
		\begin{minipage}{0.95\textwidth}
			#1 \\
			{\bf{Input:}} #2  \\
			{\bf{Task:}} #3
		\end{minipage}
	}
	\vspace{2mm}
}
\newtheorem{theorem}{Theorem}
\newtheorem{lemma}[theorem]{Lemma}
\newtheorem{proposition}[theorem]{Proposition}
\newtheorem{corollary}[theorem]{Corollary}

\title{ETH-Tight FPT Algorithm for \\ Makespan Minimization on Uniform Machines}

\author{Lars Rohwedder\footnote{\href{mailto:rohwedder@sdu.dk}{rohwedder@sdu.dk}, University of Southern Denmark, Odense, Denmark. Supported by Dutch Research Council (NWO) project “The Twilight
Zone of Efficiency: Optimality of Quasi-Polynomial Time Algorithms” [grant number OCEN.W.21.268]}}

\date{}

\begin{document}
\maketitle
\begin{abstract}
	Given $n$ jobs with processing times $p_1,\dotsc,p_n\in\N$ and $m\le n$ machines with speeds $s_1,\dotsc,s_m\in\N$ our goal is to allocate the jobs to machines minimizing the makespan.
	We present an algorithm that solves the problem in time $p_{\max}^{O(d)} n^{O(1)}$, where $p_{\max}$ is the maximum processing time and $d\le p_{\max}$ is the number of distinct processing times. This
	is essentially the best possible due to a lower bound based on the exponential time hypothesis (ETH).

	Our result improves over prior works that had a quadratic term in $d$ in the exponent
	and answers an open question by Kouteck\'y and Zink.
	The algorithm is based on integer programming techniques combined with novel ideas based
	on modular arithmetic.
	They can also be implemented efficiently for the more compact high-multiplicity instance encoding. 
\end{abstract}

\section{Introduction}
\label{sec:introduction}
We consider a classical scheduling problem, where we need to allocate $n$ jobs with processing times
$p_1,\dotsc,p_n$ to $m\le n$ machines with speeds $s_1,\dotsc,s_m$. Job $j$ takes time $p_j/s_i$
if executed on machine $i$ and only one job can be processed on a machine at a time. Our goal is to minimize
the makespan. Formally, the problem is defined as follows.

\defproblem{Makespan Minimization on Uniform Machines}{
	$n \ge m\in\N$, $p_1,\dotsc,p_n\in\N$, $s_1,\dotsc,s_m\in\N$
	}{
		Find assignment $\sigma: \{1,\dotsc,n\} \rightarrow \{1,\dotsc,m\}$ that minimizes
		\begin{equation*}
			\max_{i=1,\dotsc,m} \sum_{j : \sigma(j) = i} \frac{p_j}{s_i} \ .
		\end{equation*}
	}

The special case with $s_1 = \cdots = s_m = 1$ is called Makespan Minimization on \emph{Identical} Machines.
Either variant is strongly NP-hard and has been studied extensively towards approximation algorithms.
On the positive side, both variants admit an EPTAS~\cite{jansen2020closing, jansen2010eptas}, that is,
a $(1 + \epsilon)$-approximation algorithm in time $f(\epsilon) \cdot n^{O(1)}$
for any $\epsilon > 0$. Here, $f(\epsilon)$ is a function that may depend exponentially on $\epsilon$.

More recently, the problem has also been studied regarding exact FPT algorithms, where the
parameter is the maximum (integral) processing time $p_{\max} = \max_j p_j$ or
the number of different processing times $d = |\{p_1,\dotsc,p_n\}|$, or a combination of both.
Note that $d \le p_{\max}$. An algorithm is fixed-parameter tractable (FPT) in a parameter $k$,
if its running time is bounded by $f(k) \cdot \langle \mathrm{enc} \rangle^{O(1)}$, that is, the running time can have
an exponential (or worse) running time dependence on the parameter, but not on the overall instance encoding length $\langle \mathrm{enc} \rangle$.
The study of FPT algorithms in the context of our problem was initiated
by Mnich and Wiese~\cite{MnichW15}, who showed, among other results,
that for identical machines there is an FPT algorithm in $p_{\max}$.
The running time was improved through the advent of new generic integer programming (ILP) tools.
Specifically, a series of works led to fast FPT algorithms for highly structured integer programs called
$n$-fold Integer Programming, see e.g.~\cite{CslovjecsekEHRW21}. 
Makespan Minimization on Uniform Machines (and, in particular, the special case
on identical machines) can be modeled in this structure and one can directly derive FPT results from the algorithms
known for $n$-fold Integer Programs~\cite{KnopK18}. Namely, the state-of-the-art for $n$-fold ILPs~\cite{CslovjecsekEHRW21} leads
to a running time of
\begin{equation*}
	p_{\max}^{O(d^2)} n^{O(1)} \le p_{\max}^{O(p_{\max}^2)} n^{O(1)} \ .
\end{equation*}
Kouteck\'y and Zink~\cite{KouteckyZ20} stated as an open question whether the exponent of $O(d^2)$ can be improved to $O(d)$.
This is essentially the best one can hope for: even for identical machines
Chen, Jansen, and Zhang~\cite{chen2014optimality}
have shown that there is no algorithm that given an upper bound $U\in \N$ decides if the optimal makespan
is at most $U$ 
in time $2^{U^{0.99}} n^{O(1)}$, assuming the exponential time hypothesis (ETH).
Since $d \le p_{\max} \le U$ there cannot be an algorithm for our problem
with running time $2^{O(p_{\max}^{0.99})} n^{O(1)}$
or $p_{\max}^{O(d^{0.99})} n^{O(1)}$ either.

A similar gap of understanding exists for algorithms for integer programming in several variants, see~\cite{RohwedderW25} for an overview.
Since no improvement over the direct application of $n$-fold Integer Programming is known,
Makespan Minimization on Uniform Machines can be seen as a benchmark problem for integer
programming techniques. For brevity we omit a definition of $n$-fold Integer Programming here and refer the reader
to~\cite{CslovjecsekEHRW21} for further details.

Jansen, Kahler, Pirotton, and Tutas~\cite{JansenKPT24} proved that for the case where
the number of distinct machine speeds, that is, $|\{s_1,\dotsc,s_m\}|$, is polynomial in $p_{\max}$, the 
running time of $p_{\max}^{O(d)} n^{O(1)}$ can be achieved. Note that this includes the
identical machine case. Jansen et al.~\cite{JansenKPT24} credit a non-public manuscript by Govzmann, Mnich, and Omlo for discovering the identical machine case 
earlier and for some proofs used in their result.

\paragraph{Our contribution.}
We fully settle the open question by Kouteck\'y and Zink~\cite{KouteckyZ20}.
\begin{theorem}
Makespan Minimization
on Uniform Machines can be solved in time $p_{\max}^{O(d)} n^{O(1)}$.
\end{theorem}
We first prove this for the following intermediate problem, which
is less technically involved and simplifies the presentation of our algorithm.

\defproblem{Multiway Partitioning}{
	$n\ge m\in \N$, $p_1,\dotsc,p_n\in \N$, $T_1,\dotsc,T_m\in \N$ with
	\begin{equation*}
	p_1 + \cdots + p_n = T_1 + \cdots + T_m \ .
	\end{equation*}
}{
	Find partition $A_1\dot\cup\cdots\dot\cup A_m = \{1,\dotsc,n\}$ such that
	
	\begin{equation*}
		\sum_{j\in A_i} p_j = T_i \quad \text{for all } i=1,\dotsc,m \ .
	\end{equation*}
}

For consistency, we also use the job-machine terminology when talking about Multiway Partitioning.
As a subroutine we solve following generic integer programming problem that might be of independent interest.  

\defproblem{Multi-Choice Integer Programming}{
	$n,d,\Delta\in\N$, $A\in \Z^{d\times n}$ with $|A_{ij}| \le \Delta$, $b\in \Z^d$, $c\in\Z^n$. Further, a partition $P$ of $\{1,\dotsc,n\}$ and $t_S\in\N$ for each $S\in P$.
}{
	Find $x \in \ZZ^n$ maximizing $c\T x$, subject to $A x = b$ and $\sum_{i\in S} x_i = t_S$ for all $S\in P$.
}

\begin{theorem}\label{thm:ilp}
	Multi-Choice Integer Programming can be solved in time
	\begin{equation*}
		(m \Delta |P|)^{O(m)} (n + t)^{O(1)} \ ,
	\end{equation*}
	where $t = \sum_{S\in P} t_S$.
\end{theorem}
Our algorithm is based on an approach that Eisenbrand and Weismantel~\cite{eisenbrand2019proximity}.
introduced, where they use the Steinitz Lemma for reducing the search space in integer programming.

We note that Jansen et al.~\cite{JansenKPT24} also used a tailored integer programming algorithm to obtain
their result.
There are similarities to our ILP algorithm, which is partly inspired by it.
The method in~\cite{jansen2023integer} also reduces the search space, but
via a divide-and-conquer approach due to Jansen and Rohwedder~\cite{jansen2023integer} rather than the Steinitz Lemma.
It is the author's impression that this method may also be able to produce a guarantee similar to \Cref{thm:ilp},
but since it is not stated in a generic way, we cannot easily verify this and use it as a black box.
It seems that the approach in~\cite{JansenKPT24} suffers from significantly more technical complications than ours.
Our proof is arguably more accessible and compact.

An important aspect in the line of work on FPT algorithms for Makespan Minimization is
\emph{high-multiplicity encoding}. Since the number of possible processing times is bounded, one can
encode an instance efficiently by storing $d$ processing times $p_1,\dotsc,p_d$ and a multiplicity $n_1,\dotsc,n_d$.
Semantically, this means that there are $n_i$ jobs with processing time $p_i$. The encoding
can be much more compact than encoding $n$ many processing times explicitly. In fact, the difference
can be exponential and therefore obtaining a polynomial running time in the high-multiplicity encoding length
can be much more challenging than in the natural encoding. 
Our algorithm can easily be implemented in time $p_{\max}^{O(d)} \langle \mathrm{enc} \rangle^{O(1)}$,
when given an input in high-multiplicity encoding of length $\langle \mathrm{enc} \rangle$.
Alternatively, a preprocessing based on a continuous relaxation and proximity results
can be used to reduce $n$ sufficiently and apply our algorithm as is, see~\cite{brinkop2024}.
For readability, we use the natural instance encoding throughout most of this document.

\paragraph{Other related work.}
The special case of Multiway Partitioning where $m = 2$ is exactly the classical Subset Sum
problem. This problem has received considerable attention regarding the maximum item size
as a parameter lately. Note that in contrast to the other mentioned problems, Subset Sum is only weakly
NP-hard and admits algorithms pseudopolynomial in the number of items and the maximum size.
Optimizing this polynomial (also in the more general Knapsack setting) has been subject of a series of recent works,
see~\cite{polak2021knapsack, jin20240, chen2024faster, bringmann2024knapsack}.

It is natural to ask whether Makespan Minimization on Uniform Machines (or any of the previously mentioned variants)
admits an FPT algorithm only in parameter
$d$ (and not $p_{\max}$). In the identical machine case, this depends on the encoding type.
For high-multiplicity encoding there is a highly non-trivial XP algorithm due to Goemans and Rothvoss~\cite{goemans2020polynomiality}, that is, an algorithm with
running time $\langle\mathrm{enc}\rangle^{f(d)}$, and it is open whether an FPT algorithm exists.
For natural encoding the result of Goemans and Rothvoss directly implies an FPT algorithm,
see~\cite{KouteckyZ20}.
For uniform machines, the problem is $W[1]$-hard in both encodings, even under substantial
additional restrictions, as shown by Kouteck\'y and Zink~\cite{KouteckyZ20}.

\paragraph{Overview of techniques.}
Key to our result is showing and using the perhaps surprising fact that
feasibility of a certain integer programming formulation is sufficient
for feasibility of Multiway Partitioning. In essence, this model relaxes the load constraints for machines
with large values of $T_i$, requiring only congruence modulo $a$ for a particular choice of $a\in \N$.
We refer to \Cref{sec:main} for details.
It is not trivial to see that the model can be solved in the given time.
We achieve this via a new algorithm for Multi-Choice Integer Programming, see~\Cref{sec:ilp},
that we then use in~\Cref{sec:appl} to solve our model for Multiway Partition.
The result transfers to Makespan Minimization on Uniform Machines by a straight-forward reduction.
Finally, we sketch how to adapt the algorithm to high-multiplicity encoding in \Cref{sec:high-mult}.

\section{Modulo Integer Programming Formulation}\label{sec:main}
Our model uses a pivot element $a \in \{p_1,\dotsc,p_n\}$.
The selection of $a$ is intricate as its definition is based on the unknown solution to the problem.
We can avoid this issue by later attempting
to solve the model for each of the $d$ possible choices of~$a$.

A machine $i \in \{1,\dotsc,m\}$ is called small if $T_i < p_{\max}^4$ and big otherwise. We denote the set of small machines
by $S$ and the big machines by $B = \{1,\dotsc,m\} \setminus S$.
Define \emph{\mIP} as the following mathematical system:
\begin{align}
	\sum_{j=1}^n p_j x_{ij} &= T_i &\text{ for all } i\in S \label{eq:mod-small} \\
	\sum_{j=1}^n p_j x_{ij} &\equiv T_i \mod a &\text{ for all } i\in B \label{eq:mod-big} \\
	\sum_{j : p_j = a}\sum_{i\in B} x_{ij} &\ge p_{\max}^2 \cdot |B| & \label{eq:mod-pivot} \\
	\sum_{i=1}^m x_{ij} &= 1 &\text{ for all } j=1,\dotsc,n \label{eq:mod-assign} \\
	x_{ij} &\in \{0, 1\} &\text{ for all } j=1,\dotsc,n,\ i=1,\dotsc,m \notag
\end{align}
Here,~\eqref{eq:mod-assign} guarantees that the solution is an assignment of jobs to machines, encoded
by binary variables $x_{ij}$.
Constraint~\eqref{eq:mod-small} forces the machine load of small machines to be correct.
Instead of requiring this also for big machines, \eqref{eq:mod-big} only guarantees the correct
load modulo $a$.
Furthermore, we require that a sufficient number of jobs with processing time $a$ are assigned
to the big machines. There always exists a pivot element, for
which this system is feasible. We defer the details to \Cref{sec:appl} and dedicate the rest of this section to
proving that
any feasible solution for \mIP{} can be transformed efficiently into a feasible solution to the original problem.
In particular, feasibility of \mIP{} implies feasibility of the original problem,
regardless of the choice of $a$.

\subsection{Algorithm}\label{sec:main-alg}
\paragraph{Phase I.} Starting with the solution for \mIP{}, from each big machine we remove all jobs of processing time
$a$. Furthermore, as long as there a processing time $b\neq a$ such that at least $a$ many
jobs of size $b$ are assigned to the same big machine $i\in B$, we remove $a$ many of these jobs from $i$.
Note that both of these operations maintain Constraint~\eqref{eq:mod-big}.

However, Constraint~\eqref{eq:mod-assign} will be temporarily violated, namely some jobs are not assigned.
After the operations have been performed exhaustively, there are at most $(d-1)(a - 1) \le p_{\max}^2$ jobs on each
big machine $i$ and, using the definition of big machines, it follows that their total processing time is less than $T_i$.

\paragraph{Phase II.} We now assign back the jobs that we previously removed.
First, we take each bundle of $a$ many jobs with the same processing time that we had removed together earlier.
In a Greedy manner we assign the jobs of each bundle together to some big machine $i$, where they can be 
added without exceeding $T_i$. As we will show in the analysis, there always exists such a machine.

\paragraph{Phase III.} Once all bundles are assigned, we continue with the jobs with processing time $a$. We individually assign
them Greedily to big machines $i$, where they do not lead to exceeding $T_i$.

\subsection{Analysis}
\begin{lemma}\label{lem:bundles}
	Let $z_{ij}$ be the current assignment at some point during Phase~II.
	Then there is a big machine $i\in B$ with
	\begin{equation*}
		\sum_{j=1}^n p_j z_{ij} \le T_i - p_{\max}^2 \ .
	\end{equation*}
	In particular, adding any bundle to $i$ will not exceed $T_i$.
\end{lemma}
\begin{proof}
	Let $x_{ij}$ be the initial solution to \mIP{}, from which we derived $z_{ij}$.
	Recall that by problem definition we have $p_1 + \cdots + p_n = T_1 + \cdots + T_m$. Further,
	it holds that
	\begin{equation*}
		\sum_{i\in S} \sum_{j=1}^n p_j x_{ij} = \sum_{i\in S} T_i \ .
	\end{equation*}
	Together, these statements imply that
	\begin{equation*}
		\sum_{i\in B} \sum_{j=1}^n p_j x_{ij} = \sum_{i\in B} T_i \ .
	\end{equation*}
	Consider the operations on $x_{ij}$ that led to $z_{ij}$ and focus on one particular job $j$.
	This jobs $j$ was either removed from some big machine and added back to another big machine,
	which does not change the left-hand side of the equation above; or its assignment did not change,
	which also does not affect the left-hand side; or it was removed without being added back,
	which decreases the left-hand side by $p_j$. The latter is the case at least for the jobs
	with processing time equal to~$a$. It follows that
	\begin{equation*}
		\sum_{i\in B} \sum_{j=1}^n p_j z_{ij} \le \sum_{i\in B} T_i - \sum_{i\in B}\sum_{j : p_j = a} p_j x_{ij} \le \sum_{i\in B} T_i - p_{\max}^2 \cdot |B| \ .
	\end{equation*}
	The last inequality follows from Constraint~\eqref{eq:mod-pivot}.
	It follows that there is a big machine $i\in B$ with
	\begin{equation*}
		\sum_{j=1}^n p_j z_{ij} \le \sum_{i\in B} T_i - p_{\max}^2 \ .
	\end{equation*}
	Since each bundle consists of at most $a - 1 \le p_{\max}$ jobs with processing time at most $p_{\max}$,
	adding this bundle to $i$ will not exceed $T_i$.
\end{proof}
\begin{lemma}
	Let $z_{ij}$ be the current assignment at some point during Phase~III, but before all jobs have been assigned back.
	Then there is a big machine $i\in B$ with
	\begin{equation*}
		\sum_{j=1}^n p_j z_{ij} \le T_i - a \ .
	\end{equation*}
	In particular, a job with processing time $a$ can be added to $i$ without exceeding $T_i$. 
\end{lemma}
\begin{proof}
	With the same argument as in the proof of \Cref{lem:bundles} it follows that
	\begin{equation*}
		\sum_{i\in B} \sum_{j=1}^n p_j z_{ij} < \sum_{i\in B} T_i \ .
	\end{equation*}
	Here, we do not have the same gap as in \Cref{lem:bundles}, since some jobs of size $a$ might already be
	assigned to big machines, but we still have strict inequality, since not all jobs are assigned.
	Thus, there is a big machine $i\in B$ with $\sum_{j=1}^n p_j z_{ij} < T_i$.
	Notice that $\sum_{j=1}^n p_j z_{ij} \equiv T_i \mod a$, since the initial solution $x_{ij}$ for
	\mIP{}, from which $z_{ij}$ was derived, satisfies this and all operations we perform only add or subtract
	a multiple of $a$ from the machine load.
	It follows that
	\begin{equation*}
		\sum_{j=1}^n p_j z_{ij} \le T_i - a \ . \qedhere
	\end{equation*}
\end{proof}

\begin{theorem}
	Given a feasible solution to \mIP, the procedure described in \Cref{sec:main-alg} constructs a feasible
	solution to Multiway Partitioning in time polynomial in $n$.
\end{theorem}
\begin{proof}
	By the previous Lemmas the algorithm succeeds in finding an assignment $z_{ij}$ where each machine $i$
	has load $\sum_{j=1}^n p_j z_{ij} \le T_i$. Since
	$p_1 + \cdots + p_n = T_1 + \cdots + T_m$ equality must hold for each machine.
	The polynomial running time is straight-forward due to the Greedy nature of the algorithm.
\end{proof}

\section{Multi-Choice Integer Programming}
\label{sec:ilp}
This section is dedicated to proving \Cref{thm:ilp}.
We refer to \Cref{sec:introduction} for the definition Multi-Choice Integer Programming.

\subsection{Algorithm}\label{sec:steinitz-alg}
Let $t = \sum_{S\in P} t_S$.
On a high level we start with $x = 0$ and then for iterations $k=1,\dotsc,t$ increase a single variable by one.
We keep track of the right-hand side of the partial solution at all times. We do not, however, want to
explicitly keep track of the current progress $\sum_{i\in S} x_i$ for each $S\in P$.
Instead, we fix in advance, which set we will make progress on in each iteration, ensuring that for each $S\in P$
there are exactly $t_S$ iterations corresponding to it.
Further, we want to make sure that all sets progress in a balanced way, which will later help bound the number
of right-hand sides we have to keep track of.
For intuition, we think of a continuous time $[0, 1]$. At $0$ all variables are zero; at $1$ all
sets are finalized, that is, $\sum_{i\in S} x_i = t_S$ for each $S\in P$.
For a set $S\in P$ we act at the breakpoints $1/t_S, 2/t_S,\dotsc,t_S/t_S$.
This almost defines a sequence of increments, except that some sets may share the same breakpoints, in which
case the order is not clear.
We resolve this ambiguity in an arbitrary way.
Let $S_1,\dotsc,S_t \in P$ be the resulting sequence and $d_1,\dotsc,d_t$ the corresponding breakpoints.
Formally, we require that $d_1\le \cdots \le d_t$ and for each $S\in P$ and each $i=1,\dotsc,t_S$ there
is some $k \in \{1,\dotsc,t\}$ such that $S_k = S$ and $d_k = i/t_S$.

We now model Multi-Choice Integer Programming as a path problem in a layered graph. There are $t$ sets
of vertices $V_1,\dotsc,V_{t+1}$. The vertices $V_k$ correspond to right-hand sides $b'\in \Z^d$,
which stand for a potential right-hand side generated by the partial solution constructed in iterations $1,\dotsc,k-1$.
Formally,
$V_k$ contains one vertex for every $b'\in \Z^d$ with $\|b' - d_k \cdot b \|_{\infty}\le 4 d\Delta |P|$.
Let $v'\in V_k$ and $v'' \in V_{k+1}$ and let $b', b'' \in \Z^d$ be the corresponding right-hand sides.
There is an edge from $v'$ to $v''$ if there is some $i\in S_k$ with $A_i = b'' - b'$.
Intuitively, choosing this edge corresponds to increasing $x_i$ by one. The weight of the edge is $c_i$,
or the maximum such value if there are several $i\in S_k$ with $A_i = b'' - b'$.

We solve the longest path problem in the graph above from the $0$-vertex of $V_0$ to the $b$-vertex of $V_{t+1}$.
Since the graph is a DAG, this can be done in polynomial time in the number of vertices of the graph, which
is polynomial in $(8 d\Delta |P| + 1)^d \cdot (t+1)$.
From this path we derive the solution $x$ by incrementing the variable corresponding to each edge, as described above.

\subsection{Analysis}
It is straight-forward that given a path of weight $C$ in the graph above, we obtain a feasible solution
of value $C$ for Multi-Choice Integer Programming.
However, because we restrict the right-hand sides allowed in $V_1,\dotsc,V_{t+1}$
it is not obvious that the optimal solution corresponds to a valid path.
In the remainder, we will prove this.
\begin{lemma}\label{lem:steinitz-path}
Given a solution $x$ of value $c\T x$ for Multi-Choice Integer Programming, there exists a path of 
	the same weight $c\T x$ in the graph described in \Cref{sec:steinitz-alg}.
\end{lemma}
This crucially relies on the following result.
\begin{proposition}[Steinitz Lemma~\cite{sevast1978approximate}, see also~\cite{eisenbrand2019proximity}]
	Let $\| \cdot \|$ be an arbitrary norm.
	Let $d\in\N$ and $v_1,\dotsc,v_n\in \R^d$ with $v_1 + \cdots + v_n = 0$ and $\| v_i \| \le 1$ for all $i=1,\dotsc,n$.
	Then there exists a permutation $\sigma \in \mathcal S_n$ such that for all $i = 1,\dotsc,n$ it holds that
	\begin{equation*}
		\| v_{\sigma(1)} + \cdots + v_{\sigma(i)} \| \le d \ .
	\end{equation*}
\end{proposition}
\begin{proof}[Proof of \Cref{lem:steinitz-path}]
	Consider first one set $S\in P$.
	Let 
	\begin{equation*}
		x^{(S)}_i = \begin{cases}
			x_i &\text{ if } i\in S \\
			0 &\text{ otherwise}
		\end{cases}
	\end{equation*}
	be the solution restricted to $S$.
	For each $i\in S$ we define a vector
	\begin{equation*}
		v_i = \frac{A_i}{2\Delta} - \frac{A x^{(S)}}{2\Delta t_S} \in \R^d \ .
	\end{equation*}
	Observe that $\| v_i \|_{\infty}\le 1$ for all $i\in S$ and
	\begin{equation*}
		\sum_{i\in S} x_i v_i = \frac{1}{2\Delta}(\sum_{i\in S} A_i x_i - \frac{A x^{(S)}}{t_S} \sum_{i\in S} x_i) = \frac{1}{2\Delta}(A x^{(S)} - A x^{(S)}) = 0 \ .
	\end{equation*}
	Thus, by the Steinitz Lemma we can find a bijection $\sigma_S : \{1,\dotsc,t_S\} \rightarrow S$ such that
	\begin{equation*}
		\| v_{\sigma_S(1)} + \cdots + v_{\sigma_S(i)} \|_{\infty} \le d \quad \text{for all } i=1,\dotsc,t_S \ .
	\end{equation*}
	Using the definition of $v_i$ we obtain that
	\begin{equation*}
		\| A_{\sigma_S(1)} + \cdots + A_{\sigma_S(i)} - Ax^{(S)} \frac{i}{t_S} \|_{\infty} \le 2d\Delta \quad \text{for all } i=1,\dotsc,t_S \ .
	\end{equation*}
	Next we define the following increments: for an iteration $k\in \{1,\dotsc,t\}$ with $S := S_k$ and $d_k = i/t_S$
	we increment the variable $\sigma_S(i)$. In other words, for each set $S\in P$ we follow
	exactly the order given by $\sigma_S$.

	It remains to bound the right-hand side of each partial solution. Consider again an iteration $k\in\{1,\dotsc,t\}$.
	Let $s_S \le t_S$ be the number of increments to set $S$ that have been performed during iterations $1,\dotsc,k$. Then $s_S$ must be such that $(s_S - 1)/t_S \le d_k \le s_S/t_S$. In other words,
	\begin{equation*}
		s_S \in \{\lceil d_k t_S \rceil,\ \lfloor d_k t_S\rfloor + 1\} \ .
	\end{equation*}
	Let $x'$ be the partial solution after iteration $k$. Then using $\sum_{S\in P} Ax^{(S)} = Ax = b$ and
	several triangle inequalities, we calculate.
	\begin{align*}
		\|A x' - d_k b \|_{\infty} &= \|\sum_{S\in P} \sum_{i=1}^{s_S} A_{\sigma_S(i)} - d_k b\|_{\infty} \\
		&\le \sum_{P\in P} \|A_{\sigma_S(\lfloor d_k t_S \rfloor + 1)} \|_{\infty} + \|\sum_{S\in P} \sum_{i=1}^{\lceil d_k t_S \rceil} A_{\sigma_S(i)} - d_k b\|_{\infty} \\
		&\le |P| \cdot \Delta + \|\sum_{S\in P} \sum_{i=1}^{\lceil d_k t_S \rceil} [A_{\sigma_S(i)} - d_k \cdot Ax^{(S)}]\|_{\infty} \\
		&\le |P| \cdot \Delta + \|\sum_{S\in P} (d_k - \frac{\lceil d_k t_S \rceil}{t_S}) \cdot Ax^{(S)}\|_{\infty} \\
		&\quad + \|\sum_{S\in P} \sum_{i=1}^{\lceil d_k t_S \rceil} [A_{\sigma_S(i)} - \frac{\lceil d_k t_S \rceil}{t_S} \cdot Ax^{(S)}]\|_{\infty} \\
		&\le |P| \cdot \Delta + |P|\cdot \Delta + |P| \cdot 2d\Delta \le |P| \cdot 4d\Delta \ .
	\end{align*}
	It follows that solution $x$ can be emulated as a path $P$ in the given graph using the increment sequence
	defined above. Since the weights of the edges correspond to the values of $c$, this path has weight $c\T x$.
\end{proof}

We conclude this section by showing that the previous result extends to the case of Multi-Choice Integer
Programming with inequalities instead of equalities.
\begin{corollary}\label{cor:ilp2}
	Let $A\in \Z^{d\times n}$ with $|A_{ij}| \le \Delta$, $b\in \Z^d$, and $c\in\Z^n$.
	Let $P$ be a partition of $\{1,\dotsc,n\}$ and $t_S\in\N$ for each $S\in P$.
	In time
	\begin{equation*}
		(m \Delta |P|)^{O(m)} (n + t)^{O(1)} \ ,
	\end{equation*}
	where $t = \sum_{S\in P} t_S$, we can solve
	\begin{align*}
		\max &\ c\T x \\
		A x &\le b \\
		\sum_{i\in S} x_i &= t_S &\text{\rm for all } S\in P \\
		x &\in \ZZ^n 
	\end{align*}
\end{corollary}
\begin{proof}
We reduce to \Cref{thm:ilp} by adding slack variables.
First, we remove all trivial constraints. If $b_j \ge \Delta t$ then the corresponding constraint cannot be
violated. 
For each row $j$ of $A$ we add two variables $s_j, \bar s_j$, which
form a new set in the partition $P$ with required cardinality $s_j + \bar s_j = 2 t\Delta$.
We add $s_j$ to the left-hand side of $j$th inequality and replace it by an inequality.
Given a solution $x$ to the ILP with inequalities, we can set $\bar s_j = (b - Ax)_j$
and $s_j = 2t\Delta - s_j \ge 0$ for $j=1,\dotsc,m$ to obtain a feasible solution for the ILP with
equalities. For a feasible solution for the ILP with equalities, the same settings of variables $x$
is also feasible for the ILP with inequalities. Hence, both ILPs are equivalent.
The transformation increases $|P|$ by $d$, $n$ by $2d$, and $t$ by $2t\Delta$.
These changes do not increase the running time bound asymptotically.
\end{proof}

\section{Main Result}
\label{sec:appl}
We first need to verify that \mIP{} is indeed feasible for some choice of $a$.
\begin{lemma}
	Given a feasible instance of Multiway Partitioning, there exists a pivot $a\in \{p_1,\dotsc,p_n\}$
	such that \mIP{} is feasible.
\end{lemma}
\begin{proof}
	Consider the assignment $y_{ij}$ corresponding to the solution of Multiway Partitioning.
	By definition of the problem, this assignment satisfies Constraint~\eqref{eq:mod-assign}
	and for each $i=1,\dotsc,m$
	that
	\begin{equation}
		\sum_{j=1}^n p_j y_{ij} = T_i \ . \label{eq:sol-match}
	\end{equation}
	This implies that Constraints~\eqref{eq:mod-small} and~\eqref{eq:mod-big} are satisfied, for any choice of $a$.
	Recall that each big machine $i\in B$ has $T_i \ge p_{\max}^4$. In particular, \eqref{eq:sol-match}
	implies that 
	\begin{equation}
		\sum_{j=1}^n y_{ij} \ge \frac{T_i}{p_{\max}} \ge p_{\max}^3 \quad \text{for all } i\in B \ .
	\end{equation}
	Thus,
	\begin{equation}
		\sum_{j=1}^n \sum_{i\in B} y_{ij} \ge p_{\max}^3 \cdot |B| \ .
	\end{equation}
	Thus, there exists some $a\in \{p_1,\dotsc,p_n\}$ with
	\begin{equation}
		\sum_{j : p_j = a} \sum_{i\in B} y_{ij} \ge \frac{p_{\max}^2 |B|}{d} \ge p^2_{\max} \cdot |B| \ .
	\end{equation}
	In other words, Constraint~\eqref{eq:mod-pivot} holds for this choice of $a$, which concludes the proof.
\end{proof}
We will now model the problem of solving \mIP{} as an instance
of Multi-Choice Integer Programming.
The following is a relaxation of \mIP{}:
\begin{align*}
	\sum_{j=1}^n p_j x_{ij} &= T_i &\text{ for all } i\in S \\
	\sum_{j=1}^n p_j x_{ij} &\equiv T_i \mod a &\text{ for all } i\in B \\
	\sum_{j : p_j = a}\sum_{i\in S} x_{ij} &\le |\{j \mid p_j = a\}| - p_{\max}^2 \cdot |B| & \\
	\sum_{i=1}^m x_{ij} &\le 1 &\text{ for all } j=1,\dotsc,n \\
	x_{ij} &\in \{0, 1\} &\text{ for all } j=1,\dotsc,n,\ i=1,\dotsc,m
\end{align*}
Here, we swap the constraint on jobs of size $a$ to the small machines instead of the large ones
and, more importantly, we do not require all jobs to be assigned.
This model and \mIP{} are in fact equivalent, since all jobs that are unassigned must have a total processing time
that is divisible by $a$ (because of $p_1 + \cdots + p_n = T_1 + \cdots T_m$ and the constraints). Thus,
one can derive a solution to \mIP{} by adding all unassigned jobs to an arbitrary big machine, assuming
$B\neq \emptyset$.
If, on the other hand, $B = \emptyset$ then the requirement that small machines have the correct load implies that all
jobs are assigned, making the model exactly equivalent to \mIP{}.

We can therefore focus on solving the model above, which is done with the help of the standard modeling technique 
of \emph{configurations}. 

Notice that in the model above small machines can have at most $p_{\max}^4$ jobs assigned
to each. For big machines, we may also assume without loss of generality that at most $(a-1)d \le p_{\max}^4$ jobs are
assigned to each, since otherwise we can remove $a$ many jobs of the same processing time without affecting feasibility. 

Further, there are only a small number of machine \emph{types}:
for the small machines there are only $p_{\max}^4$ possible values of $T_i$ and all machines
with the same value of $T_i$ behave in the same way; for big machines, all machines with the
same value of $T_i \mod a$ behave the same and thus there are only $a$ many types.
For one of the $p_{\max}^4 + a$ many types $\tau$, we say that
a vector $C\in \ZZ^d$ is a configuration
if the given multiplicities correspond to a potential job assignment, namely
$\sum_{k = 1}^d p_j C_j = T(\tau)$ if $i\in S$ and $\sum_{k = 1}^d p_j C_j \equiv T(\tau) \mod a$ if $i\in B$. Here, $T(\tau)$ is the target (or remainder modulo $a$) corresponding to the type $\tau$.
We denote by $\mathcal C(\tau)$ the set of configurations for type~$\tau$ and by $m(\tau)$ the number of
machines of type~$\tau$. In the following model, we use variables $y_{\tau, C}$ to describe how
many machines of type $\tau$ use configuration $C\in\mathcal C(\tau)$.
\begin{align*}
	\sum_{C\in \mathcal C(\tau)} y_{\tau, C} &= m(\tau) &\text{for all types } \tau \\
	\sum_{\tau \text{ small}}\sum_{C\in \mathcal C(\tau)} C_a \cdot y_{\tau, C} &\le |\{j\mid p_j = a\}| - p^2_{\max} \cdot |B| & \\
	\sum_{\tau}\sum_{C\in \mathcal C(\tau)} C_b \cdot y_{\tau, C} &\le |\{j\mid p_j = b\}| &\text{for all } b\in \{p_1,\dotsc,p_n\} \\
	y_{\tau, C} &\in \ZZ &\text{ for all } \tau, C
\end{align*}
It is straight-forward that this model is indeed equivalent to the previous one.
The integer program has the structure of Multi-Choice Integer Programming (with inequalities) partitioned into
the sets
$\{y_{\tau, C} \mid C\in C(\tau)\}$ for each type $\tau$.
The maximum coefficient of the constraint matrix is $p_{\max}^4$, the number of rows of constraint matrix $A$
is $d + 1$, and the sum of cardinality requirements $t$ is $m$.
Applying \Cref{cor:ilp2} this leads to a running time of $p_{\max}^{O(d)} m^{O(1)}$, assuming
the values $|\{j\mid p_j = b\}|$ have been precomputed.

\paragraph{Makespan Minimization on Uniform Machines.}
We use a binary search framework to the problem, where given $U\in\RR$
our goal is to determine if there is a solution $\sigma$ with
machine loads $\sum_{j : \sigma(j) = i} p_j \le T_i := \lfloor s_i U \rfloor$ for each machine~$i$.
Since the optimal value has the form $L / s_i$ for some $i\in{1,\dotsc,m}$ and $L\in\{0,1,\dotsc,n p_{max}\}$,
a binary search all these values increases the running time by a factor of only 
$O(\log(m \cdot n \cdot p_{\max}))$, which is
polynomial in the input length.
Our techniques rely heavily on exact knowledge of machine loads. To emulate this,
we add $T_1 + \cdots + T_m - p_1 - \cdots - p_n$ many dummy jobs 
with processing time $1$. Clearly, this maintains feasibility and, more precisely,
creates a feasible instance of Multiway Partitioning, assuming that $U$ is a valid upper bound.
Note that $d$ may increase by one, which is negligible with respect to our running time.
We can now solve the resulting instance using the algorithm for Multiway Partitioning.

\section{High-Multiplicity Encoding}
\label{sec:high-mult}
Recall that in the high-multiplicity setting we are given $d$ processing times $p_1,\dotsc,p_d$
with multiplicities $n_1,\dotsc,n_d$ (next to the machine speeds $s_1,\dotsc,s_d$). The encoding length
is therefore
\begin{equation*}
	\langle \mathrm{enc} \rangle = \Theta(d \log(p_{\max}) + d\log(n) + m \log(s_{\max})) \ .
\end{equation*}
A solution is encoded by vectors $x_{ij} \in \ZZ$ that indicate how many jobs of processing time $p_j$
are assigned to machine~$i$, which is of size polynomial in $\langle \mathrm{enc} \rangle$.
Through a careful implementation we can solve
Makespan Minimization on Uniform Machines
also in time $p_{\max}^{O(d)} \cdot \langle \mathrm{enc} \rangle^{O(1)}$.
The binary search explained at the end of \Cref{sec:appl} adds only an overhead factor
of $O(\log(n m p_{\max})) \le \langle \mathrm{enc} \rangle^{O(1)}$
Notice that the ILP solver in \Cref{sec:appl} already runs in that time
of $p_{\max}^{O(d)} \cdot m^{O(1)}$, which is sufficiently fast. This needs to be repeated
$d$ times for each guess of $a$.
Afterwards, we need to implement the Greedy type of algorithm in \Cref{sec:main}.
Instead of removing one bundle or job at a time, we iterate over all machines and processing times and
remove as many bundles as possible in a single step.
This requires only time $O(md)$. Similarly, we can add back bundles and jobs
of size $a$ in time $O(md)$ by always adding as many bundles as possible in one step.

\bibliographystyle{plain}
\bibliography{ref}

\begin{thebibliography}{10}

\bibitem{bringmann2024knapsack}
Karl Bringmann.
\newblock Knapsack with small items in near-quadratic time.
\newblock In {\em Proceedings of {STOC}}, pages 259--270, 2024.

\bibitem{brinkop2024}
Hauke Brinkop, David Fischer, and Klaus Jansen.
\newblock Structural results for high-multiplicity scheduling on uniform
  machines.
\newblock {\em CoRR}, abs/2203.01741, 2024.

\bibitem{chen2014optimality}
Lin Chen, Klaus Jansen, and Guochuan Zhang.
\newblock On the optimality of approximation schemes for the classical
  scheduling problem.
\newblock In {\em Proceedings {SODA}}, pages 657--668, 2014.

\bibitem{chen2024faster}
Lin Chen, Jiayi Lian, Yuchen Mao, and Guochuan Zhang.
\newblock Faster algorithms for bounded knapsack and bounded subset sum via
  fine-grained proximity results.
\newblock In {\em Proceedings {SODA}}, pages 4828--4848, 2024.

\bibitem{CslovjecsekEHRW21}
Jana Cslovjecsek, Friedrich Eisenbrand, Christoph Hunkenschr{\"{o}}der, Lars
  Rohwedder, and Robert Weismantel.
\newblock Block-structured integer and linear programming in strongly
  polynomial and near linear time.
\newblock In {\em Proceedings of {SODA}}, pages 1666--1681, 2021.

\bibitem{eisenbrand2019proximity}
Friedrich Eisenbrand and Robert Weismantel.
\newblock Proximity results and faster algorithms for integer programming using
  the steinitz lemma.
\newblock {\em ACM Transactions on Algorithms (TALG)}, 16(1):1--14, 2019.

\bibitem{goemans2020polynomiality}
Michel~X Goemans and Thomas Rothvo{\ss}.
\newblock Polynomiality for bin packing with a constant number of item types.
\newblock {\em Journal of the ACM (JACM)}, 67(6):1--21, 2020.

\bibitem{jansen2010eptas}
Klaus Jansen.
\newblock An eptas for scheduling jobs on uniform processors: using an milp
  relaxation with a constant number of integral variables.
\newblock {\em SIAM Journal on Discrete Mathematics}, 24(2):457--485, 2010.

\bibitem{JansenKPT24}
Klaus Jansen, Kai Kahler, Lis Pirotton, and Malte Tutas.
\newblock Improving the parameter dependency for high-multiplicity scheduling
  on uniform machines.
\newblock {\em CoRR}, abs/2409.04212, 2024.

\bibitem{jansen2020closing}
Klaus Jansen, Kim-Manuel Klein, and Jos{\'e} Verschae.
\newblock Closing the gap for makespan scheduling via sparsification
  techniques.
\newblock {\em Mathematics of Operations Research}, 45(4):1371--1392, 2020.

\bibitem{jansen2023integer}
Klaus Jansen and Lars Rohwedder.
\newblock On integer programming, discrepancy, and convolution.
\newblock {\em Mathematics of Operations Research}, 48(3):1481--1495, 2023.

\bibitem{jin20240}
Ce~Jin.
\newblock 0-1 knapsack in nearly quadratic time.
\newblock In {\em Proceedings {STOC}}, pages 271--282, 2024.

\bibitem{KnopK18}
Dusan Knop and Martin Kouteck{\'{y}}.
\newblock Scheduling meets n-fold integer programming.
\newblock {\em Journal on Scheduling}, 21(5):493--503, 2018.

\bibitem{KouteckyZ20}
Martin Kouteck{\'{y}} and Johannes Zink.
\newblock Complexity of scheduling few types of jobs on related and unrelated
  machines.
\newblock In {\em Proceedings of {ISAAC}}, volume 181, pages 18:1--18:17, 2020.

\bibitem{MnichW15}
Matthias Mnich and Andreas Wiese.
\newblock Scheduling and fixed-parameter tractability.
\newblock {\em Mathematical Programming}, 154(1-2):533--562, 2015.

\bibitem{polak2021knapsack}
Adam Polak, Lars Rohwedder, and Karol W{\k{e}}grzycki.
\newblock Knapsack and subset sum with small items.
\newblock In {\em Proceeding of {ICALP}}, pages 1--19, 2021.

\bibitem{RohwedderW25}
Lars Rohwedder and Karol Wegrzycki.
\newblock Fine-grained equivalence for problems related to integer linear
  programming.
\newblock In {\em Proceedings of {ITCS}}, 2025.

\bibitem{sevast1978approximate}
S~Sevast’janov.
\newblock Approximate solution of some problems of scheduling theory.
\newblock {\em Metody Diskret. Analiz}, 32:66--75, 1978.

\end{thebibliography}

\end{document}